\documentclass[reqno,11pt]{article}

\usepackage[english]{babel}
\usepackage{bbm} 
\usepackage[latin1]{inputenc}
\usepackage{amsmath,amsthm,amsfonts,amssymb}
\numberwithin{equation}{section}
\usepackage[final]{epsfig}
\usepackage{color}

\usepackage{hyperref}

\usepackage{txfonts}

\usepackage{hyperref}

\usepackage{mathrsfs}

\usepackage{enumerate}

\def\titlerunning#1{\gdef\titrun{#1}}
\makeatletter
\def\author#1{\gdef\autrun{\def\and{\unskip, }#1}\gdef\@author{#1}}
\def\address#1{{\def\and{\\\hspace*{18pt}}\renewcommand{\thefootnote}{}%
\footnote {#1}}%
\markboth{\autrun}{\titrun}}
\makeatother

\def\subjclass#1{{\renewcommand{\thefootnote}{}%
\footnote{\emph{Mathematics Subject Classification (2010):} #1}}}

\newtheorem{theorem}{Theorem}[section]
\newtheorem{corollary}[theorem]{Corollary}
\newtheorem{lemma}[theorem]{Lemma}

\newtheorem{proposition}[theorem]{Proposition}

\theoremstyle{definition}
\newtheorem{definition}[theorem]{Definition}

\numberwithin{equation}{section}

\newcommand{\overbar}[1]{\mkern 1.5mu\overline{\mkern-1.5mu#1\mkern-1.5mu}\mkern 1.5mu}

\usepackage{graphicx,pgf}

\makeatletter
\let\Im\undefined
\let\Re\undefined
\DeclareMathOperator{\Im}{Im \,}
\DeclareMathOperator{\Re}{Re \,}

\DeclareMathOperator{\tr}{tr}

\newcommand{\Z}{ {\mathbb Z} }
\newcommand{\N} {{\mathbb N}}
\newcommand{\C}{\mathbb{C}}

\newcommand{\R}{\mathbb{R}}

\def\Ev#1{{\mathbb E}\left(#1 \right)}
\def\Pr#1{{\mathbb P}\left(#1 \right)}

\def\T{{\mathcal T}} 
\usepackage{bbm}
\def\id {\mathbbm{1} }

\DeclareMathOperator{\indfct}{\rm 1}

 \def\XXint#1#2#3{{\setbox0=\hbox{$#1{#2#3}{\int}$}
 \vcenter{\hbox{$#2#3$}}\kern-.5\wd0}}

\def\be{\begin{equation}}
\def\ee{\end{equation}} 
\def\Pr#1{ {\mathbb P}{\left(   #1 \right)}}

\def\={\  = \  }   
\def\+{\  + \  }   
\def\-{\  - \  }    
\def\:{\  : \   }    

\makeatother

\begin{document}




\titlerunning{On the ubiquity of the Cauchy distribution in spectral problems}

\title{On the ubiquity of the Cauchy distribution in spectral problems} 

\author{Michael Aizenman  \and  Simone Warzel
}

\date{July 31, 2014}

\maketitle

\address{M. Aizenman: Departments of Physics and Mathematics,  Princeton University, USA.
\and S. Warzel: Zentrum Mathematik, TU M\"unchen, 
 Boltzmannstr. 3, 85747 Garching, Germany.}



\begin{abstract} 
We consider the distribution of the values at real points of random functions which belong to the  Herglotz-Pick   (HP) class of analytic mappings of the upper half plane into itself.  
It is shown that under mild stationarity assumptions the individual values of HP 
functions with singular spectra have a  Cauchy type distribution.    
The statement applies to the diagonal matrix elements of random operators, and holds regardless of the presence or not of   level repulsion, i.e. applies to both random matrix and  Poisson-type spectra.



\end{abstract}


\setcounter{tocdepth}{2} 
{\small
\tableofcontents 
}

\section{Introduction}

In the study of spectral properties of random operators, generically denoted below by $H_\omega$,  one is led to consider random elements of  
a class of functions of a complex variable $z$, which is variably named after Pick (\cite{Do74}) or  Herglotz  (\cite{dk05}).
Included in this class are  functions of the form: 
\begin{eqnarray}  \label{eq:trace}
R_{\omega,n}(z) &=&  \frac{1}{n} \tr \frac{1}{H_{\omega,n} -z} \=  \frac{1}{n} \sum_{j=1}^n \frac{1}{E_{\omega,j}^{(n)}-z} \notag \\ 
\mbox{} \\ 
 \quad   R_{\omega,n}^{\phi}(z) & = &   \langle \phi, \frac{1}{H_{\omega,n}-z} \, \phi\rangle    \= \sum_{j=1}^n \frac{|\langle \phi | \psi_{\omega,j}^{(n)}\rangle |^2}{E_{\omega,j}^{(n)}-z}  \,.  \notag 
\end{eqnarray}  
where $H_{\omega,n}$ are operators acting in spaces of finite dimension, 
or alternatively  $n\times n$ matrices.   In the second example   $\phi$ is a vector in the  space on which $H_{\omega,n}$ acts and the expressions on the right correspond to the operator's spectral representation.  
\\

More generally,   the Herglotz-Pick (HP) class, as defined here\footnote{In a variant of the definition the range of the functions is occasionally restricted to  $ \mathbb{C}^+ $.  Its extension here to   $\overline{\mathbb{C}^+}$   allows to  include  the  degenerate  case of functions of constant real value.},  
consists of analytic functions from the upper half plane $\mathbb{C}^+ :=  \{ z \in \C \, | \, \Im z > 0 \} $   into its closure $\overline{\mathbb{C}^+} = \mathbb{C}^+ \cup{\R}$. 
By the Herglotz
 representation theorem (cf.~\cite{dk05}) each such function admits a unique spectral representation  as 
 \be  \label{eq:F} 
F(z) \=  b+ a z  + \int_\R \left(  \frac{1}{u-z} - \frac{u}{u^2+1}\right) \mu(du)\, . 
\ee
with $a \ge 0$, $b \in \R$, and $\mu$ a non-negative Borel measure on $ \R $, which is referred to as the spectral measure of $ F $, for  which: 
\be \label{eq:HN_norm}
 \int(u^2+1)^{-1} \mu(du)   < \infty   \,. \\  
 \ee

%
%

Of particular interest here are the scaling limits  in which spectra of finite dimensional operators of increasing dimension are studied on a scale of the eigenvalue spacing~\cite{Seba90,AlMa98, BGS1, BBK,KMW, FR05,ASW}.   The functions of interest may be found to converge in an appropriate distributional sense to random HP functions of singular spectrum which in the simplest case consists of simple poles located  along $\R$.  In the latter case, 
 \eqref{eq:F} extends to a random meromorphic functions, whose spectra form a random point process on $\R$.  \\

Our main purpose here is two fold.  One is to  clarify some of the     
relations between  shift invariant scaling limits of point processes and the limits of the corresponding  random HP functions.   The other is to present the general observation that 
 translation invariance, and more specifically `shift amenability', of an HP function with singular spectral measure carries the implication that  the probability  distribution of the boundary values $F(x) := F(x+i0)$ is a Cauchy distribution.   The examples to which this principle applies include scaling limits of   eigenvalue point processes of a number of random matrix models where the spectral statistics are of otherwise quite different nature.   This includes  both limits of  random diagonal matrices without level repulsion, and those of random matrix ensembles within the GXE domains of attraction.   The latter case includes a class of random Wigner matrices for which the result is established through a combination of the general criteria derived here with previous analytical results derived in \cite{ESY08,TV11,EY12} on the convergence of the local law  to the scaling limit of the GUE ensemble.  \\

 The topics discussed here are of relevance for  
 quantum transport in mesoscopic quantum systems.  In that context, an argument for the general appearance of the Cauchy distribution was  first presented   by P.A.\ Mello 
 \cite{Mel94}~\footnote{We thank Y.V. Fyodorov for alerting us to the references.},  
 who proposed also an  extension of this principle to a somewhat less universal law concerning the limiting  (joint) distribution of arbitrary size ($k\times k$) resolvent subblocks  of   random matrices of much larger size ($n \times n$, with $n \gg k$)).   Support for some of Mello's reasoning was presented  by P.W. Brouwer \cite{Brow95}, who pointed out that also the statement is strictly true  within a Lorentzian matrix ensemble, where it holds for any $k\le n$, and in \cite{FS}, \cite[Ch.IV]{FY1} and \cite[App.~A]{FY2} using supersymmetric calculations on other GXE ensembles in the large $n$ limit.  Other, more recent results are mentioned in Section~\ref{sec:examples}.

\section{Cauchy distribution in shift amenable HP functions}  \label{sec:Cauchy} 

\subsection{Definition and main result}

It is of relevance to recall here the following general result.

\begin{proposition}[de la Vall\'ee Poussin, see eg.\ \cite{dk05,Duren}]  \label{thm:dlVP}
For any function $ F(z)$ 
in the HP class the limit  
\be 
F(x+i 0)\ := \  \lim_{\eta \downarrow 0}  F(x+i\eta)  
\ee 
exists for Lebesgue - almost every $x\in \R$.    
\end{proposition} 

\begin{definition} 
A   measurable function  $K: \R \mapsto \C $ will be said here to be {\it shift amenable}  if  there is a probability measure $\nu$ on $\C$ (supported necessarily on  its range's closure $\overbar {\text{Ran }   K} $)  such that  for any continuous  bounded  function $\Psi :  \overbar {\text{Ran }   K} \mapsto \C$  the following limit exists and satisfies
\be \label{Llim}
\lim_{L\to \infty} \frac{1}{L} \int_{-L/2}^{L/2} \Psi(K(x) ) \, dx \ = \  \int_\C \Psi( w ) \, \nu(dw)  \  =: \  \nu(\Psi) \,.  
\ee 
We refer to $\nu\equiv\nu_K$ as  {\it $K$'s distribution under shifts}. 
\end{definition}

In other words, a function is shift amenable if when sampled uniformly over the range $[-L/2,L/2]$, with $L\to \infty$,  the distribution of the values of $K(x)$ is asymptotically described by a probability measure $\nu$ on $\overbar {\text{Ran }   K}$.  \\ 

As it is noted in Section~\ref{subset:stationary}
shift amenable functions appear  naturally among the {\it typical realizations} of random functions with shift invariant law.    The following statement is however deterministic in the sense that it applies to every shift amenable function.

%
  
\begin{theorem} \label{thm:Cauchy}
Let $F(z)$ be a HP function whose boundary values satisfy: 
\begin{enumerate} 
\item 
 $\Im F(x+i0) \= 0 \qquad \mbox{ for Lebesgue almost every $x\in \R$.}  $     
\item  $F_0(x):=F(x+i0)$ \qquad is shift amenable.   
\end{enumerate} 
Then under shifts $F_0(x)$  has a  Cauchy distribution.  
\end{theorem}

By a Cauchy distribution we refer here to a probability law, parametrized by $\Gamma \in \overline{\mathbb{C}^+}$,  of the form
\be  \label{eq:Cauchy}  
\Pr{dF} \ = \ \pi^{-1} \frac{\Im \Gamma\, dF}{(F-\Re \Gamma)^2+ (\Im \Gamma)^2} \, ,  
\ee 
which for $\Gamma \in \R$ is to be interpreted as a $\delta$-measure located at $\Re \Gamma$.  
 We refer to   $\Gamma \in \C^+$ as the Cauchy distribution's `analytic baricenter'.   
More is said on its value in the present context in Theorem~\ref{lem:Aconst} below.\\

Condition~1.) in Theorem~\ref{thm:Cauchy}  is equivalent to the statement that the spectral measure $\mu$ of $ F $ has no absolutely continuous component, as the latter is in general  given by $\pi^{-1} \Im F(x+i0) \, dx$.       In the theorem's proof use is made of  the following  auxiliary statements.  
In the first one, the range of $K$  is limited to  $\overbar \C^+ $  in order to make the statement applicable to functions such as $\Psi(z) =  1/(z+i)$.

\begin{lemma}  \label{lem:g} 
If  $K$  is a {\it shift amenable} function with range $\text{Ran } K = \overbar \C^+ $ then  for any 
bounded continuous function $\Psi : \overbar \C^+ \mapsto \C$, and any monotone decreasing   $g:\R_+ \mapsto \R_+$   
satisfying the  normalization condition $\int_\R g(|u|) \, du =1$: 
  \be 
\nu_{K}(\Psi)
 = \lim_{\eta\to \infty}  \int_\R \Psi(K(u)) \, g(|u-x|/\eta) \, \frac{du}{ \eta } \,  ,   \ 
 \ee 
 where the limit does not depend on $x\in \R$.  
\end{lemma} 
\begin{proof}   
For $g(x) = 2^{-1}   \id[|x| <1]$  the statement holds by  the definition of shift amenability.    The extension to more general $g$ is by a standard application of Abel's lemma, which can be deduced through  the `layer-cake' representation:  
$g(t) = \int_0^\infty  \id[ g(t) \ge \tau]\,  d \tau $.    
\end{proof} 
 
\begin{lemma} \label{lem:lim} 
Let $ F(z): \C^+ \mapsto \overline{\mathbb{C}^+} $  be a  Herglotz - Pick function whose boundary value function $F_0(x):= F(x+i0)$  
is shift amenable.   
Then the following limits exist and satisfy:
\begin{enumerate}
\item  for any  bounded continuous $\Psi : \overline{\C^+} \to \C $ which is analytic on $ \C^+$, and any $x\in \R$:  
\be \label{eq:name} 
  \nu_{F_0}(\Psi)  \ = \ \lim_{\eta \to \infty} \Psi(F(x+i\eta)) \,  , 
\ee 
\item  
for  every $x\in \R$ (which however does not affect the limit):   
\be \label{eq:Gamma2} 
\lim_{\eta \to \infty} F(x+i\eta) \=  \left\{  \int \frac{\nu_{F_0}(dw)}{w +i}   \right\}^{-1} -i \  =: \  \Gamma  \, . 
\ee 
\end{enumerate} 
 \end{lemma} 
 \begin{proof}
 1.~Since $ z \mapsto \Psi(F(z)) $ is bounded and analytic over $\mathbb{C}^+$, its values where $\Im z>0$  admit the Poisson integral representation (cf.~\cite[Thm.~11.2]{Duren}):     
 \be  \label{eq:useof_g}
 \Psi(F(x+i\eta )) \=   \int \Psi(F(u+i 0)) \, \frac{\pi^{-1} \, \eta \, \,  du}{(u-x)^2+\eta^2}   \,,   \ 
 \ee  
 By Lemma~\ref{lem:g}, with $g(u) = \pi^{-1}/(u^2+1)$, in the limit $\eta \to \infty$ the expression on the right converges to   $\nu_{F_0}(\Psi)$.  \\ 
 
\noindent
2.~ The second statement, \eqref{eq:Gamma2}, follows by applying  \eqref{eq:name} to the function $\Psi(w) := -[w+i]^{-1}$.\\ 
\end{proof}

\begin{proof}[{\bf Proof of Theorem~\ref{thm:Cauchy} }] 
By \eqref{eq:name}, applied to the function $\Psi(w) = e^{it w }$ with $t \in (0,\infty)$, we learn that:  
\be \label{eq:tmoment}
\int e^{i t w}  \nu_{F_0}(dw) \ =\  
\lim_{\eta\to \infty}   e^{i t F(x+i\eta)}   \  = \ e^{it\Gamma}
\ee 
where the limit is evaluated using  \eqref{eq:Gamma2}.  \\

The above argument yields the generating function of the probability measure $\nu_{F_0}$ for $t>0$ (that part being applicable regardless of the first assumption of the theorem).   However,  
under the assumption that $F_0(x)$ is a.s. real for $x\in \R$ the generating function at $t<0$ can also be obtained from \eqref{eq:tmoment} through complex conjugation.    Thus, under this assumption, for any $t\in \R$: 
\be \label{char}
 \int e^{i t z}  \nu_{F_0}(dz)  \ = e^{it \Re \Gamma - |t| \Im \Gamma} \,. 
\ee  
Since probability measures on $\R$  are uniquely determined by their characteristic functions, \eqref{char}  implies  that the probability distribution $\nu_{F_0}$ coincides with that of  $ \Re \Gamma +  \xi\,  \Im \Gamma $
 where  $\xi$ is the standard Cauchy random variable of the probability distribution $\pi^{-1} d \xi/[\xi^2+1]$. 
\end{proof} 
 
\medskip  

\subsection{Examples and the relation of the Cauchy law with Boole's identity}
 
Following are some examples of functions to which  Theorem~\ref{thm:Cauchy} applies.  One may note that these functions  differ quite significantly in the structure of the higher correlations, which however  do not affect the common Cauchy law.

\begin{enumerate} 
\item The periodic function (cf.~\cite[Ch.~19]{Herglotz_trick})
\be  F^{Per} (z) \= -\pi \cot(\pi z) \= \frac{-1}{z} \-
    \sum_{n=1}^\infty \left( \frac{1}{z-n} +  \frac{1}{z+n} \right) \, \\ 
    \ee 
\item  Quasi-periodic  functions of the form
\be 
F^{QP}(z) \=  - \sum_{j=1}^M  \alpha_j \cot (\beta_j z+ \theta_j) 
\ee 
with $\alpha_j \ge 0$ and  $\beta_j, \theta_j \in \R$.  
\item  The random function with Poisson distributed poles: 
\be \label{eq:Poisson}
F_\omega^{Poi}(z) \= \lim_{N\to \infty} \sum_{u\in \omega \cap [-N,N]} \frac{1}{u-z}
\ee 
where $\omega \subset \R$ is a random configuration of the Poisson point process on $\R$ with intensity $dx$.   In this case, the assumptions of Theorem~\ref{thm:Cauchy} hold for almost every $\omega$ (cf. Subsection~\ref{eq:St_examples}).  
\item A function whose singularities have the $\beta$-ensemble statistics, e.g. 
\be  \label{eq:GUE}
F_\omega^{GUE}(z)  \= \lim_{N\to \infty} \sum_{u\in \omega \cap [-N,N]} \frac{1}{u-z} 
\ee 
where $\omega$ is a configuration of the shift invariant determinantal point process associated with the kernel  $K(x,y) = \frac{\sin\pi (x-y)}{\pi(x-y)}$  (cf. Subsection~\ref{eq:St_examples}).
\item  More generally than the previous two examples, $F(z)$ could be a random HP  function of shift invariant distribution, as defined in Subsection~\ref{subset:stationary} below.   The almost-sure shift amenability of such functions   is the consequence of Birkhoff's ergodic theorem.   
\end{enumerate} 

The universality of the first order statistics, which holds regardless of the differences in the second order statistics expresses the fact that  the fraction of the Lebesgue measure ($\mathcal L$): 
\be 
\frac{ \mathcal L (\{ x\in [-L/2,L/2] \, : \  F(x) >t \} )}{L} 
\ee  
is not affected by a wide range of rearrangements of the singularities.   These may include both shifts of the singularities locations and splits of their mass.    A similar ``integrability''   condition can be spotted to lie behind   an identity  which G. Boole presented  to the Royal Society in 1857.  In a slightly generalized form, the Boole identity may be stated as follows.

\begin{proposition} [Extension of Boole~\cite{Boole}]  \label{prop:boole} 
For any finite singular measure $\mu(dx)$ which has no absolutely continuous component the function 
\be 
F(z) \= \int_\R \frac{\mu(du) } {u-z} 
\ee
satisfies, for all for any $t>0$:
\be
\mathcal L (\{ x\in \R \, : \  F(x+i0) >t \} )  \=   \frac{\mu(\R)}{t} 
\ee
\end{proposition}  
Boole's original Theorem was stated and proven for point measures of finite support.  For convenience, a  proof of this generalization is enclosed in Appendix~\ref{app:Boole}. 

\section{The spectral representation and related topology} 
\subsection{An alternative spectral representation}

As an alternative to \eqref{eq:F},  each HP function can also be written as   
\be \label{eq:circrep}
F(z) = G(\, w(z)\, ) \qquad \mbox{with} \quad z \ =\  i \frac{1+w}{1-w}  ,\qquad   w\ = \  \   \frac{z-i}{z+i}\, , 
\ee 
with 
\be \label{eq:G} 
G(w) \ = \   b\ + \ \int_{S} \sigma( d \theta) \,  \,  i\, \,  \frac{e^{i\theta} +w}{e^{i\theta} -w}  \,  
\ee 
where $ \sigma$ is a uniquely defined finite measure on the unit circle $S$ and $ w $ is a point in the unit disk $ \mathbb{D} $.   
The   correspondence between the two representations is:  
\be\label{eq:meascorresp}
\frac{\mu(dx)}{x^2+1} \ = \ \sigma( d\theta)  \,  \id[\theta \neq 0]\, , \qquad a \ = \  \sigma(\{0 \}) \, . 
\ee 
with $x= - \cot(\theta/2)$, where the coefficient $a$ of \eqref{eq:F}  is incorporated  as a $\delta$-point mass of $\sigma$.  \\

Thus, any  HP function is uniquely associated with the pair $\left (\sigma, G(0) \right)$ in the space
\be \label{eq:Omega} 
\Omega \ =\  \left\{ \sigma \in \mathcal{M}(S) \, :  
\int  \sigma(d\theta)    < \infty \, \right\}\,  \times\,   \overline {\mathbb{C}^+ }\, \
\ee 
or equivalently  with the pair $(\mu, F(i) )  $ in the space
\be 
 \widetilde  \Omega \  =\   \left\{  \mu \in \mathcal{M}(\R)   \, \:  
\int \frac{\mu(du) }{u^2+1}  < \infty \, \right\}   \,  \times \,  \overline {\mathbb{C}^+} \, ,  
\ee 
and correspondingly the space of HP functions can be identified with either  $ \Omega$ or $\widetilde  \Omega $, and we shall be frequently switching between the two.   \\


\subsection{The topology of pointwise convergence }  \label{sec:pointwise_conv}

A natural topology on the collection of  HP functions is that of uniform convergence on  compact subsets of $ \C^+ $ (uniform convergence preserves  analyticity as well as the restriction $\Im F(z) \ge 0$). 
However, it is a  known consequence of the Montel theorem that  under an added restriction on the range of the functions the conditions can be simplified.  In particular, for HP functions uniformity on compacta follows from just pointwise convergence over $\C^+$ (as can also be seen from the   bounds presented below). 
In this section our goal is to clarify the expression of this topology  in terms of the spectral representation.  Particularly  convenient for this purposes is the  
representation of HP functions in  the space $\Omega$, in terms of  \eqref{eq:G}.  \\  

The parameter $b$, as well as the total  mass of the measure $\sigma(S) $ are continuous in the topology of pointwise convergence, since  
\be \label{eq:normalization}
b \= \Re G(0)\, , \qquad  \sigma(S) = \Im G(0) \, , 
\ee  
where $G(0) \equiv F(i)  $.  \\

For a pair of  measures on $S$  the variational distance is  
\be 
|m_1 -m_2| = \sup  \left\{ \int_{S} \   f(e^{i\theta})\  
 [ m_1(d\theta)-m_2(d\theta) ] \ \big| \  f \in C(S) ; \ \|f\|_{ \infty} \le 1  \right \} \, , 
 \ee 
  and in case of measures  of equal mass  the  Wasserstein distance  is
\be 
W(m_1,m_2) \= \sup \left\{ \int_{S} \   f(e^{i\theta})\  
 [ m_1(d\theta)-m_2(d\theta) ] \ \big| \   \rm{Lip}(f) \le 1 \right \} \, .   
\ee      
To make this applicable to pairs of HP functions $G_j$ with spectral measures $\sigma_j(d\theta)$ of different masses,  
let us  first consider the case   $ \sigma_j(S) \neq 0 $ and denote  the corresponding probability measures  by
\be 
\widetilde \sigma_j (d \theta) \=  \frac{\sigma_j(d\theta)}{\sigma_j(S)} \ .   \\ 
\ee 
By direct estimates,    
for all $w\in \mathbb D$ and all $ \theta\in S$,
\be  \label{eq:WZbounds} 
 \left|  \left(\frac{e^{i\theta}+w}{e^{i\theta}-w} \right) \right| \ \le \  \frac{2}{ 1-|w|} \, ,   \qquad 
  \left| \frac{d}{d\theta} \left(\frac{e^{i\theta}+w}{e^{i\theta}-w} \right) \right| \ \le \  \frac{2\, |w|}{ (1-|w|)^2} \, ,    
\ee
one may  hence conclude: 
\begin{multline}   \label{eq:Gbound}  
|G_1(w) \ - \ G_2(w)|  
\  \le \  |\Re  G_1(0) - \Re  G_2(0) | \  + \  \left| \Im \left(  G_1(0) -  G_2(0) \right) \right|  
\, \frac{2 }{ 1-|w|}   \\[1.5ex]   
\+ \frac{\sigma_1(S) + \sigma_2(S)}{2} \,  
W\left(\widetilde \sigma_1, \widetilde \sigma_2\right) \, 
\frac{2\, |w|}{ (1-|w|)^2} \, .     
\end{multline}  
If  one (or both) of the measures is of  zero mass the last term can be dropped, since then its  ``normalized'' measure can be selected arbitrarily, and ``by fiat'' it can be arranged so that $ \widetilde \sigma_j $ are equal and thus  $W\left(\widetilde \sigma_1, \widetilde \sigma_2\right)  =0$.\\ 

 These bounds are of help  in establishing the following  equivalence.

\begin{theorem} \label{prop:topology}
For a sequence of HP functions $G_n: \mathbb D \mapsto\overline{\mathbb{C}^+} $  the following are equivalent:
\begin{enumerate}
\item[A.] The pair of conditions:    
   \begin{enumerate}
   \item[1.]  the single-site  limit exists: $ \lim_{n\to \infty} G_n(0)  =: G(0) $. 
   \item[2.]   the  spectral measures  $\sigma_n$ on $S$ (defined by \eqref{eq:G}) converge weakly to a measure $\sigma \in \mathcal M(S)$, in the sense that 
   $  \sigma_{n}(g) \to \sigma(g)  $ for every continuous $g\in C(S)$.
   \end{enumerate}
\item[B.] There exists a HP function $G$ such that for all $w \in \mathbb D$:  $G_n(w) \to  G(w) $,  uniformly on compact subsets of $ \mathbb D $.  
\item[C.] The functions $G_n$ converge pointwise over $\mathbb D$. 
\end{enumerate}
 \end{theorem}

 \begin{proof}  
\noindent $``A \Rightarrow\, B'': $    Set $b=\Re G(0)$, and let $G: \mathbb D \mapsto \overline{\C^+}$ be the function which corresponds to $(\Re G(0), \sigma)  $ under \eqref{eq:G}.    (The definition is consistent with the previously determined $G(0)$, since under the assumption [A1-A2],  the  condition $ \Im G_n(0) \= \sigma_n(S)$  persists also in the limit $n\to \infty$.)  \\ 

 The claim that $[G_n(w) - G(w) ]\to 0$ uniformly on compact subsets of $\mathbb D$ will be verified separately for two cases:\\ 
i.   $\sigma(S) \=0$:  \  The claim follows from \eqref{eq:Gbound} (without the last term) and $\Im G_n(0) = \sigma_n(S)\to \sigma(S) = 0$.   \\ 
 ii.   $\sigma(S) \neq 0$: \   In this case the weak convergence of the measures implies  that also the normalized measures converge weakly, and by implication  also  in the Wasserstein distance.  Thus 
 \be 
 \lim_{n\to \infty} W\left(\widetilde \sigma_n, \widetilde \sigma\right) \= 0 \,. 
 \ee 
 The claim then follows from \eqref{eq:Gbound}. \\

\noindent $``B \Rightarrow\, C'': $  is evident.\\ 
 
\noindent  $``C\Rightarrow\, A'': $  The convergence of $G_n(0) $ directly implies  [A1].  By \eqref{eq:normalization} this  implies convergence of $b_n$ as well as that of the total mass $\sigma_n(S)$.  
  
   By the compactness of the set of measures on $S$ with $\sigma(S) \le \Im G(0)$, the sequence $\sigma_n( d \theta)$ has accumulation points, to which it converges over suitable subsequences $(n_k)$.   All these measures share the values of the following integrals: 
   \be \label{eq:diffrep}
 \int_{S}  \frac{\sigma( d \theta)}{e^{i\theta} -w}  \= \lim_{k\to \infty} 
\int_{S}  \frac{\sigma_{n_k}( d \theta)}{e^{i\theta} -w}  \= \frac{1}{ 2i w} \ \lim_{k\to \infty}  ( G_{n_k}(w) - G_{n_k}(0) ) \, ,    
 \ee
for all $w \in \mathbb{D}\backslash \{0\} $ (in addition to  $w=0$ which was established already).  
A standard argument~\cite{CFKS}, based on the  Stone-Weierstrass theorem (and resolvent identities), allows to conclude that: i. along each such subsequence the measures converge weakly, ii. the limit is uniquely characterized by \eqref{eq:diffrep}, and thus $\sigma_n$ is a convergent sequence. \\ 
 \end{proof}

\noindent {\bf Remarks on Theorem~\ref{prop:topology} :}
\begin{enumerate}  
\item  The equivalence $ B \Leftrightarrow C $ is a known consequence of the more general Montel  theorem.  

\item For  [A1] the  point $0$ is convenient, but with a minor adjustment in the argument it can be replaced by any other (pre selected) $w_0 \in \mathbb D$.   

\item   The statement can  be alternatively expressed in terms of the functions  $F_n\equiv G_n \circ w$ which are defined over $\C^+$, and the corresponding spectral measures $\mu_{n}$  on $\R$.   The main difference is   that   [A2] is to replaced by the condition:  \\[2ex] 
{\it { [A2']}   The measures $\mu_{n}(f)$ converge vaguely, in the sense  that  
$  \mu_{n}(f) \to \mu(f)  $ for all continuous, compactly supported functions  $ f \in C_c(\R)$.  }   \\[1.5ex] 
In terms of $\sigma(d\theta)$, the condition [A2']  corresponds to vague convergence on $S\backslash \{0\}$, which is a weaker statement that [A2] (since that guarantee the preservation of the total mass $\sigma(S)$).  
The two are however equivalent under the assumption [A1],  since [A2] may be concluded from [A2'] plus the convergence $ \sigma_n(g) \to \sigma(g) $ of a single function $g\in C(S)$ with $g(1) \neq 0$.  \\ 
  \end{enumerate}

In view of the rather direct correspondence between the representations of HP functions  as  $ G_{\omega,n} : \mathbb D \mapsto \overline{\C^+}$ versus $ F_{\omega,n} = G_{\omega,n}  \circ w : \C^+ \mapsto \overline{\C^+}$,   from here on we shall not be duplicating  the various statements of interest and instead use the language which locally appears to be convenient.  \\

  It may be worth noting that in contrast to $b$, the parameter 
\be
a \ = \Im F(i) - \int_\R \frac{\mu_F(dx)}{x^2+1} \= \sigma(\{0\})
\ee 
 is not  a continuous function on $\Omega$.  That is  clearly seen in the circle representation, where it corresponds to the fact that weak convergence of measures on $\C$ allows for the build  up of a $\delta$-function at $\{1\}$.  
 
\section{Stieltjes transforms of random measures} \label{sec:Stieljes}

\subsection{A constructive criterion} 

Spectral measures of interest often come in the form of random Borel measures, $ \mu_\omega $ on $\R$ (a concept discussed e.g. in ~\cite{Kal}), with the indexing parameter $\omega$ ranging over  a probability space $ (\Omega, \mathcal{A}, \mathbb{P} ) $ over which we have the action of the  group of shifts of $\R$, represented by  measurable transformations   $\{ \mathcal{T}_a \}_{a\in \R}$ for which  $  \mu_{T_a \omega}$ coincides with the shifted measure $ T_a \mu_{\omega}$,  the action of shifts on measures being defined by: 
\be \label{shift}
 \T_a \mu (I) = \mu (I+a) \,. \\ 
 \ee     

The following deterministic result presents conditions under which  the Stieltjes transform  may be constructed for such measures as the pointwise  limit, $F_{\mu}(z) := \lim_{n\to \infty}  F_{\mu}^{(n)}(z)$, of 
\be\label{eq:Stieltjes}
 F_{\mu}^{(n)}(z) := \int_{-n}^n \frac{\mu(dx)}{x-z} \, .   
\ee
It is worth noting that under the conditions listed there the functional $\mu \mapsto  F_\mu$ is shift covariant, even though this property  may at first be questioned since the ``principal value''-like  integral seen in \eqref{eq:Stieltjes} is centered at $x=0$.\\ 

In the statement we compare the Stieltjes transform of $\mu$ with a reference measure $\overbar \mu$, which in applications to random measure may be the mean value of $\mu$ averaged over that source of randomness.   Let $N_\mu(x)$ be the counting function, and    $\delta N(x)$ the difference (which in the above example corresponds to  the fluctuating part) defined by: 
\be  \label{defN}
N_\mu(x) := \int_0^x \mu(dy) \, , \qquad 
N_{\overbar \mu}(x) := \int_0^x  \overbar \mu (dy) \, , \qquad 
\delta{N}(x) \ :=\  N_\mu(x)  - N_{\overbar \mu}(x)    \, . 
\ee
Through integration by parts:
\be\label{eq:repStieltjes}
F_{\mu}^{(n)}(z) = F_{\overbar \mu}^{(n)}(z)   \ +\  \left[ \frac{\delta  N(n)}{n-z} \ - \    \frac{\delta  N(-n)}{-n-z} \right] \  +\   \int_{-n}^n  \frac{\delta  N(x)}{(x-z)^2} \, dx \,. 
\ee
Using this representation  one has the 
the following  criterion for the existence of the Stieltjes transform.

\begin{theorem}\label{thm:Stieltjes}
Let $ \mu $ and $\overbar \mu$ be a pair  of  Borel measures on $\R$ with the properties:
\begin{enumerate}[i.]   
 \item for the reference measure
  the following limit exists for all (or equivalently, by Theorem~\ref{prop:topology}, for some) $z\in \C^+$ : 
 \be \label{eq:assconvmubar}
  \lim_{n\to \infty}  \int_{-n}^n \frac{\overbar \mu(dx)}{x-z} \ =: \  F_{\overbar \mu} (z) \, ,  
 \ee
\item the difference in the pair's counting functions, defined by  \eqref{defN}, satisfies:   
\be\label{eq:assStieljes}
\lim_{n\to \pm \infty} \frac{\delta  N(n)}{n} = 0 \, \qquad \mbox{and} \qquad  \int  \frac{|\delta  N(x)|}{x^2+1} \, dx  < \infty \, .    
\ee
\end{enumerate}
Then:
\begin{enumerate}
\item   the limit (to which we refer as the Stieltjes transform)
\be\label{eq:Stieltjes2}
F_{\mu}(z) := \lim_{n\to \infty} F_{\mu}^{(n)}(z)
\ee
exists for all   $ z \in \C^+ $. 
\item for each $t \in \R$:
\be\label{eq:limitetainfty}    \lim_{\eta \to \infty} \left[ F_\mu(t+i\eta) \ - \  F_{\overbar \mu}(t+i\eta) \right ] \ = \ 0 \,. 
\ee
\end{enumerate} 
 If in addition
  \be \label{eq:massmuoverbar}
\lim_{|n|\to \infty}  \frac{\overbar \mu([n,n+1])}{n} \ = \ 0
\ee
then 
\begin{enumerate}
\item[3.]
 the resulting Stieltjes transform   is a shift-covariant functional of $\mu$, in the sense that  
the limit in \eqref{eq:Stieltjes2} exists also for $\mu$  replaced by any of the  shifted measures defined by 
\eqref{shift} and  for all $ a \in \R $ and $ z \in \C^+ $ 
\be 
 F_{\mathcal{T}_a \mu}( z ) = F_{\mu}(z+a) \, . 
 \ee  
 \end{enumerate}
 \end{theorem}

\begin{proof} 
1.~ Since the truncated measures $\id[-n,n]\, \mu(dx) $  converge to $\mu$ in the vague topology,  by Theorem~\ref{prop:topology}   the limit~\eqref{eq:Stieltjes2} exists or not simultaneously for  all $ z \in \C^+ $, and hence it suffices to test the convergence at  $ z = i $.  

Applying the representation~\eqref{eq:repStieltjes}  at  $ z = i $, the first term on the right converges by~\eqref{eq:assconvmubar}. The second and third terms converge almost surely to zero by the first assumption in~\eqref{eq:assStieljes}. The integral in the forth term is absolutely convergent, which ensures the convergence of this term. \\

2.~From the first part of this proof and~\eqref{eq:repStieltjes}  we learn that for any $ z \in \mathbb{C}^+ $:
\be
F_\mu(z) \ - \  F_{\overbar \mu}(z) \ = \ \int_{\mathbb{R}} \frac{\delta N(x) }{(x-z)^2} dx \, . 
\ee
Monotone convergence implies $ \lim_{\eta \to \infty } \int | \delta N(x) | /[(x-t)^2+\eta^2] dx = 0 $ for any $ t \in \mathbb{R} $ and hence the claim~\eqref{eq:limitetainfty}. \\

3.~In order to establish the shift-covariance we note that it is straightforward to show that for all $ n \in \N $, $ a \in \R $ and $ z \in \C^+ $:
\be
	 F_{\mathcal{T}_a \mu}^{(n)}(z) = F_{\mu}^{(n)}(z+a) + \int_n^{n+a} \frac{\mu(dx)}{x-a-z} - \int_{-n}^{-n+a} \frac{\mu(dx)}{x-a-z} \, . 
\ee
Each of the two terms on the right side converge to zero as $ n \to \infty $. This is seen through the representation
\be
\int_n^{n+a} \frac{\mu(dx)}{x-z} = \int_n^{n+a}\frac{\overbar \mu(dx)}{x-z} + \frac{\delta  N(n+a)}{n+a-z} -  \frac{\delta  N(n)}{n-z} +  \int_n^{n+a} \frac{\delta  N(x)}{(x-z)^2} dx 
\ee
(and analogously for the second term). The first term goes to zero as $ n \to \infty $ by~\eqref{eq:massmuoverbar}. 
The remaining three terms converge to zero using~\eqref{eq:assStieljes}. 
\end{proof}

\subsection{A pair of examples}   \label{eq:St_examples}

The criterion of Theorem~\ref{thm:Stieltjes} apply in particular to the following  examples of random spectral measures on $\R$, which are rather different nature.      
\begin{enumerate}
\item {\bf Poisson-Stieltjes function:}~In this example $ \mu_\omega $ is a Poisson process with constant intensity $1 $.   Picking for the reference   measure $\overbar \mu(dx) = dx $  (the Lebesque measure), one finds that  for every $ \varepsilon > 0 $:
\be 
|  \delta N_\omega(x) | \leq C_\omega(\varepsilon) \, \left( |x|^{\frac{1}{2} + \varepsilon} + 1 \right) \, . 
\ee
for all $x\in R$, with  $C_\omega(\varepsilon) $ which is almost surely finite.  
Consequently, the assumptions~\eqref{eq:assStieljes} are almost surely met in this case.   We refer to the function defined by the corresponding limit \eqref{eq:Stieltjes} as the Poisson-Stieltjes function.
\item {\bf The sine-kernel Stieltjes function:} The determinantal point process with the kernel $ K(x,y) = \sin(\pi  (x-y))/ [\pi (x-y)]$ defines an ergodic random measure $ \mu_\omega $ whose intensity is $1 $ (cf.~\cite{Soshnikov,Cupbook}).   To verify~\eqref{eq:assStieljes} for this case, with $\overbar \mu (dx) = dx$, one may use the observation that by an explicit computation (cf.~\cite[Ex.~4.2.40]{Cupbook}), for $ |x| \to \infty $:  
\be
\mathbb{E}\left[\delta N (x)^2 \right] = \int_0^x K(s,s) ds - \int_0^x  \int_0^x K(s,t) \, d s dt = \frac{\log\left( |x|\right)}{\pi^2} + \mathcal{O}(1) \, .
\ee
Consequently,  also in this case the integral in  \eqref{eq:assStieljes} is absolutely convergent.   Moreover, a Chebychev estimate  shows that 
$ \sum_n \mathbb{P}\left(  |\delta N(n)|/|n| > \varepsilon  \right) < \infty $ for any $ \varepsilon > 0 $, and hence, by the Borel-Cantelli  lemma, also the first condition in \eqref{eq:assStieljes} is met.
We refer to the function  defined through the limit \eqref{eq:Stieltjes}  with $\mu$ corresponding to this process  as the sine-kernel Stieltjes function $ F^{GUE}_\omega $.  
 \end{enumerate}

\noindent In both cases the  measure are  stationary and even ergodic.  
The  functions which are defined through the almost-sure limit~\eqref{eq:Poisson} provide  examples of  random stationary HP functions, a term to whose further exploration we turn next.   
It should be added that a construction related to \eqref{eq:assStieljes} was studied (for the Poisson process) in \cite{AlSeba}.     
However, the approach presented  there breaks the  shift covariance.

\medskip

\section{Random HP functions}  \label{sec:rand_HP}


Standard considerations imply that the 
function space $\Omega$ (and equivalently  $\widetilde \Omega$)   whose topology is discussed in Section~\ref{sec:pointwise_conv} is  metrizable 
and can be presented as a complete separable metric space (cf.~\cite{Kal}).    Estimates which are somewhat similar to \eqref{eq:Gbound} (though less explicit) are facilitated by  the ``flat metric'' (c.f. \cite{Dudley,Piccoli_Rossi}): 
\be 
d(\sigma_1,\sigma_2) \= \inf_{\widehat \sigma_1, \widehat \sigma_2 \in \mathcal M(S);  
|\widehat \sigma_1|=| \widehat \sigma_2|}  
\left[ \, |\sigma_1 - \widehat \sigma_1| \+  |\sigma_1 - \widehat \sigma_1|  \+ W( \widehat \sigma_1, \widehat \sigma_2) \,   \right] \,. \\ 
\ee

\begin{definition} \label{def:distconv}  (Random HP functions) \\ 
1.  Denoting by $\mathcal B$  the  Borel $\sigma$-algebra on $\Omega$ which corresponds to the  topology discussed above, a {\it random Herglotz-Pick function}  is given by a probability measure on  $(\Omega, \mathcal B)$.        

2.  A sequence of random HP function $ F_{\omega,n} : \C^+ \mapsto \overline{\C^+}$ is said to converge in distribution to $ F_{\omega} $ iff
 the probability measure on $\widetilde \Omega = \mathcal M(\R) \times \overline{\mathbb{C}^+}$ which forms the distribution of  $ ( \mu_{\omega,n} ,  F_{\omega,n}(i)   ) $  converges (weakly) to that of the $  ( \mu_{\omega} ,  F_{\omega}(i)   ) $.   Such convergence will be denoted $ ( \mu_{\omega,n}(f) ,  F_{\omega,n}(i)   ) \stackrel{\mathcal{D}}{\to} \left(   \mu_{\omega} , F_{\omega}(i)\right)  $. 
 \end{definition}
 
 By general theory of probability measures on complete separable metric spaces, of a finite diameter,  the  convergence of measures on $\Omega$ is equivalent to the condition that for any $\varepsilon >0$, there is $N(\varepsilon)<\infty$ such that for all $n \ge N(\varepsilon)$ the measures can be coupled so that:
 \be  \label{eq:coupling} 
 \int  d \mu_n(  \omega, \omega')  \left[ |G_\omega(0)-G_{\omega'}(0)| \+ d({\sigma_\omega},{\sigma_{\omega'}}) \right] \  \le \  \varepsilon
 \ee 
 with the marginals of $\mu_n(  \omega, \omega') $ yielding the distributions of $F_{\omega} $  and 
$ F_{\omega'} $,  correspondingly.  
(In case the distance function is unbounded, \eqref{eq:coupling} is to be replaced by the statement that  $[...]$ is small in probability, though possibly not in the mean.) \\

For future purpose let us also add
\begin{lemma} \label{lem:hole}
Let  $ F_{\omega,n} : \C^+ \mapsto \overline{\C^+}$ be a sequence of random HP functions which converges in distribution to a random HP function $ F_{\omega} $, and for which the support of the spectral measures stays away from an interval $[a, b] \subset \R$, in the sense that for some $\varepsilon >0$, and almost all $ \omega$ and all  $ n$:
\be 
\mu_{\omega,n}([a-\varepsilon,b+\varepsilon]) \ = \ 0 \,. 
\ee  
Then   the functions $F_{\omega,n}$ and $F_{\omega}$  are (almost surely)  analytic and real along $[a,b]$, and the convergens in distribution extends to:     $ ( \mu_{\omega,n},  F_{\omega,n}(x)   ) \stackrel{\mathcal{D}}{\to} \left(   \mu_{\omega} , F_{\omega}(x)\right)  $ for any $x\in [a,b]$.  
\end{lemma} 
\begin{proof}  The analyticity of the functions $ F_{\omega,n}$ within  spectral gaps is a simple consequence of  the spectral representation.  Analyticity at $x\in [a,b]$ allows to applying the harmonic average principle to the analytic continuation of $F_{\omega,n}$ through the spectral gap which includes $[a,b]$,  by which:
\be 
F_{\omega,n}(x) \= \int_{[0,2\pi]} F_{\omega,n}(x+ e^{i\theta} \varepsilon / 2)\; \frac{d \theta}{2\pi}  
\ee 
The convergence in distribution then readily follows from the coupling estimate \eqref{eq:coupling} and the {\it uniform} pointwise bound \eqref{eq:Gbound} (and the observation that the analytic continuation of such a HP function into $\C^-$ is given by the natural extension  of the spectral representation to that regime).   
\end{proof}

\subsection{Translation invariance and its consequences} \label{subset:stationary}

\subsection{Shift invariance and shift amenability}

In the  language of probabilistic ergodic theory, the subject may be presented in the following terms.      
 
%

Random functions are parametrized by a variable $\omega$ taking values in  a probability space 
 $ (\Omega_0, \mathcal{A}, \mathbb{P} ) $   (for which  a possible choice  is itself the suitable space of functions such as $\Omega$ discussed above).  The random functions are given by a $\C$-valued kernel $K_\omega(x)$  defined over 
 $ \Omega \times \R$   such that  $K_\omega(x)$ is jointly measurable over $\Omega \times \R$ (to which may optionally be added topological properties, such as discussed above).  
Translation invariance, or the more limited invariance under discrete shifts, is expressed in the two additional properties: 
 \begin{enumerate}    
\item  acting on $\Omega_0$ is  a group of measurable mappings $\{\mathcal{T}_u\}_{u\in \R}$ 
which provides a representation of the group of translations of $\R$, with
\be  \label{shift_def}
K_{T_u \omega} (x) = K_{\omega} (x+u)  \qquad  \mbox{(for almost every $(\omega,x)$)} \,, 
\ee 
\item  the probability measure $\mathbb{P} $  is invariant under  the action of the shifts $\mathcal{T}_u$, or at least under the action  of a discrete sub group $\{\mathcal{T}_{n\tau} \}_{n\in \Z}$ of period $\tau$.  
\end{enumerate}

In the above setup, let $\ell_\omega(dw)$  be the pullback measure of the conditional distribution of the values of  $K_\omega(x)$ with $x$ averaged with the Lebesgue measure over $ [0,\tau]$    at given $\omega$.  In other words,  $\ell_\omega(dw)$   is defined so  that for each continuous bounded function $\Psi: \C \mapsto \C$ 
\be 
\int_\C \Psi(w) \ell_\omega(dw) \ = \  \int_{0}^{ \tau} \Psi(K_\omega(x)) \, dx \ =: \  A_{\Psi,K}(\omega) \, .
\ee 
The average seen in \eqref{Llim} can be presented through the relation: 
\begin{eqnarray}  \label{212}
\frac{1}{L} \int_{0}^{L \tau} \Psi(K(x) ) \, dx \  = \  \frac{1}{L}  \sum_{n=0}^{L-1} A_{\Psi,K}(\mathcal{T}_{n\tau} \omega)   \, ,
\end{eqnarray}    
Birkhoff's ergodic theorem allows then to conclude that for $\mathbb {P}$-almost every $\omega$ 
 the function $K_\omega$ is shift-amenable (over $x$).  
 Furthermore, 
$
\nu_{K_\omega} \ = \ \tau \lim_{L \to \infty} \frac{1}{L}  \sum_{n=0}^{L-1}   \ell_{\mathcal{T}_{n\tau}  \omega}  \, , 
$  
and as is easily seen: 
\be 
\nu_{K_\omega} = \nu_{K_{T_\tau}\omega} 
\ee
This implies also that in the ergodic case $\nu_{K_\omega} $ is almost surely given by a common measure (on $\overbar {\text {Range }  K}$).


In a slight abuse of notation we shall generically use the symbol $\mathcal{T}_u$ for translations corresponding to shifts of $\R$, i.e. for both the transformations on $\Omega_0$ and for their induced actions on functions and measures on $\R$.   
The explicit form of this mapping in the representation of random $HP$ functions which was introduced in Section~\ref{sec:rand_HP}, 
for which $\Omega_0 = \widetilde \Omega$ and $\omega = (\mu,F(i)) $ is easily seen to take the form of the {\it co-cycle} evolution:   
\be\label{eq:transonpair}
\mathcal{T}_u (\mu,\beta) \= \left(\mathcal{T}_u\mu, \, \beta  + Q(u,\mu)  \+  u\, a   \right)
\ee 
with 
\be 
a := \Im \beta - \int \frac{\mu(dx)}{x^2+1} \, , \qquad Q(u,\mu)  \= \int\left[ \frac{1}{x-u - i} - \frac{1}{x-i} \right]  \mu(dx)   \,, 
\ee 
and  $\mathcal{T}_u \mu$ the usual shift of measure, i.e. $\left( \mathcal{T}_u \mu\right)(I) =\mu(I+u) $ for every bounded Borel set $ I \subset \R $.

Focusing on this case we take as definition: 
\begin{definition}  
 A   probability measure on $\widetilde \Omega$  is  \emph{stationary} (or {\it translation invariant}) if and only if it is invariant under the mapping induced on it by the above defined mappings
  $\{ \mathcal{T}_u \}_{u \in \R}$.  \\
  Equivalently, we will refer to the corresponding random HP function as stationary.  
\end{definition}

The following observation makes the results of Section~\ref{sec:Cauchy} applicable to stationary HP functions.

\begin{lemma} \label{lem:Psi} 
Let $ F_\omega $ be a stationary random HP function, and $F_{0,\omega}(x) := F_\omega(x+i0)$.  Then: 
\begin{enumerate}
\item With probability one the function $F_{0,\omega}(x)$ is shift amenable and the corresponding measures $\nu_{F_{0,\omega}}$ are constant on ergodic components of the probability measure.    
\item    For each bounded continuous $\Psi : \overline{\C^+} \to \C $ which is analytic on $ \C^+$: 
\be\label{202a}   \mathbb{E}\left[ \Psi(F(z)) \right] \=  \mathbb{E}\left[ \Psi(F(x+i0)) \right] 
\ee 
for all $z\in \C^+$ and $x\in \R$. 
\end{enumerate} 
 \end{lemma} 
 \begin{proof}
The first assertion   is readily implied by Birkhoff's ergodic theorem, as is explained above \eqref{212}. 

For the second, we note that by translation invariance $ \mathbb{E}\left[  \Psi(F(z)) \right] $ does not depend on $x := \Re z$.   Since under the assumptions it  forms an analytic function of $z\in \mathcal{C}^+$, it follows that it also does not depend on $y= \Im z$.   One may then deduce \eqref{202a} using Proposition~\ref{thm:dlVP} and applying the dominated convergence theorem to the limit $y\downarrow 0$ .  \\
\end{proof} 

 Thus Theorem~\ref{thm:Cauchy} is applicable to such functions.    Let us note also the following implications of stationarity.

\begin{theorem} \label{lem:Aconst}
Let $F_\omega$ be a   stationary random HP function for which $\Im F(x+i0) = 0$ almost surely (separately at each $x\in \R$).  Then:
\begin{enumerate} 
\item  for almost all $ \omega $:  
$ 
a_\omega =0  \, .  
$ 
\item  If the process is also  ergodic then for each $x\in \R$  the random variable $F_\omega(x+i0)$ has the Cauchy distribution of width: 
 \be\label{eq:imconv}
\Im \Gamma \ \equiv \  \lim_{\eta\to \infty}  \Im F_\omega(x+i\eta)   \= \pi \rho \, ,  
 \ee  
where $\Gamma$ is the distribution's analytic baricenter, as defied in \eqref{eq:Cauchy}, and 
$\rho = \Ev{\mu_F([0,1)}$.
\end{enumerate}    
\end{theorem}
\begin{proof}
1.   For any $x\in \R$ and $t>0$: 
\be   
\mathbb{P}(a =0 ) \=  \mathbb{E}\left( \id [ a=0 ] \right) \ \ge \  \lim_{y\to \infty} 
\mathbb{E}\left[ e^{itF(x+i y)} \right]  \=   
\mathbb{E}\left[ e^{itF(i )} \right]
 \, ,
 \ee 
where the inequality is by the observation that $\Im F(x+iy) \ge a y$, and the last equality by Lemma~\ref{lem:Psi}
Taking now $t\to 0$ we conclude (applying again the dominated convergence theorem):
\be
\Pr{a =0} \ \ge \  \lim_{t \to 0} 
\mathbb{E}\left[ e^{itF( i )} \right] \= 1 \,. \\ 
\ee  

2.    By Birkhoff's theorem, for ergodic processes averages over $\omega$ yield (almost surely) the same result as averages over shifts,  and thus the Cauchy nature of the distribution follows from 
Theorem~\ref{thm:Cauchy}.  The value of the  distribution's analytic baricenter  is determined from the spectral representation \eqref{eq:F} applying Lemma~\ref{lem:g}  
with $g(u) = \pi^{-1}/(u^2+1)$ as in \eqref{eq:useof_g}:    
 \be\label{eq:imconv2}
\Im \Gamma \= \lim_{\eta\to \infty}  \Im F_\omega(x+i\eta)  = \lim_{\eta\to \infty}  \frac{\pi}{\eta} \int g(|u-x|/\eta) \mu_\omega(du) \ \stackrel{a.s.}{=} \   \pi \rho \, .  \\ 
 \ee
 \end{proof}

\noindent{\bf Remarks:} 
 \begin{enumerate}
 \item The center $ \Re \Gamma \in \R $ of the Cauchy distribution of $F_\omega(x+i0)$ is not determined from the spectral measure alone, since      
 adding a real constant to a random, ergodic HP function produces another such  function with  a different value of this parameter. 
 \item Inspecting the above proof shows that one may exchange the assumption of ergodicity in the above theorem by requiring i)~stationarity of the HP function together with ii)~the distributional convergence $ F(i\eta) \to\Gamma $ with some $ \Gamma\in \C^+ $. 
 \end{enumerate}
 
 \medskip

\subsection{A cocycle criterion}

  Clearly, for any shift invariant random HP function $F_\omega$ the spectral measure $\mu_\omega$ forms a stationary random measure on $\R$, which in the discrete case corresponds to a point process.    One may ask about the converse direction: under what conditions would a random measure on $\R$  with a translation invariant distribution (and a.s. satisfying \eqref{eq:HN_norm}) be  the spectral measure  of a stationary random HP function? \\  

It is easy to see that  \eqref{eq:HN_norm}  suffices for the association with $\mu_\omega$  of the function 
\be 
K_\omega(z) \=   \int \frac{\mu_\omega(dx) }{(x-z)^2}  \, , 
 \ee 
 which is  holomorphic over $ \C^+$ and which inherits  the stationarity of $\mu_\omega$.   
The above question can therefore be rephrased as asking under what conditions would $K_\omega(z)$ be the derivative of a stationary random HP function.   For that a standard ergodic theory argument is of relevance.  \\  
 
  \begin{theorem}   \label{prop:cocycle}
  Let $\mu_\omega$ be a stationary random measure on $\R$ (given by a measurable function from a probability space to $\mathcal M (\R)$), satisfying (almost surely) \eqref{eq:HN_norm}.   
  Then $\mu_\omega$ may be extended to the spectral measure of a random stationary HP function if and only if the cocycle  
  \be 
   \Re Q(u,\mu_\omega)  =  \Re \int  \left[\frac{1}{x-u - i} - \frac{1}{x-i} \right]  \mu_\omega(dx) 
  \ee 
  is tight.  That is, if and only if: 
  \be \label{eq:Retight}
  \lim_{t\to \infty} \,  \sup_{u\in \R}  \,  
 \Pr{  \left| \Re  \int\left[\frac{1}{x-u - i} - \frac{1}{x-i} \right]  \mu(dx) \right |\, > \, t }   \= 0 
  \ee 
  \end{theorem} 
  
 \begin{proof} 
By a general result in ergodic theory~\cite{Schm} the tightness condition \eqref{eq:Retight}  allows to conclude  that the cocycle is a coboundary, i.e. there exists a measurable map $ b: \widetilde \Omega \to \R $ such that for all $ u \in \R $:
\be
 \Re Q(u,\mu_\omega)  = b_{\mathcal{T}_u\omega} - b_{\omega} \, . 
\ee
The HP function given by
\be\label{eq:HSform2}
F_\omega(z) =  b_{\omega} + \int \left[ \frac{1}{x-z} - \frac{x}{x^2+1} \right] \mu_\omega(dx) 
\ee
is then i)~almost surely well defined by~\eqref{eq:HN_norm} and ii)~easily seen to be stationary.

Conversely, if the random HP function $ F_\omega $ is stationary, it is of the form~\eqref{eq:HSform2} with $ b_\omega =   \Re F_\omega(i) $ and 
$ \Im F_\omega(i) = \int  \frac{\mu_\omega(dx)  }{x^2+1} $  (by Theorem~\ref{lem:Aconst}).  
Therefore $ Q(u,\mu_\omega) = F_\omega(i+u) - F_\omega(i) $ forms a tight collection of random variables indexed by $ u \in \R $.  
  \end{proof} 
  
  \medskip 
  
\section{Convergence criteria for the scaling limit of random HP functions}

HP functions often appear as the resolvent functions of random hermitian $n\times n$ matrices $H_{\omega,n}$ for which it is of interest to gain understanding of the  behavior of the spectra in the limit $n\to \infty$.  Examples were given in~\eqref{eq:trace}.
If the norm $\|H_{\omega,n}\|$ remains  uniformly bounded, the relevant spectra  consist of $n$ points whose gaps may typically be of order $O(n^{-1})$.  
To study this function at that level of resolution in the vicinity of an energy $E_0$ (which in principle could also depend on $n$, or be randomized in the vicinity of a target value), it is natural to enquire about the possible convergence in distribution of the random HP functions
\be   \label{eq:FTrace}
F_{\omega,n}( z) \  := \   R_{\omega,n}(E_0+z/n)  \, . \\ 
\ee  

\subsection{Convergence off the real axis} 

For  the examples~\eqref{eq:trace}, the convergence of the distribution of the spectral measure $ \mu_{\omega,n} = \sum_j \delta_{ x_{ \omega,n}(j)} $ in essence is a local statement about the behavior of the function's singularities.   The  information which is added to it through   $F_{\omega,n}(i)$   reflects the local effect of the   tails of the spectral measure, which affect its Stieltjes transform in the vicinity of $E_0$.   In many cases 
 of interest one may expect the  
contribution from the ``distant'' values of the spectrum to have only asymptotically vanishing fluctuations.   In such situations, the following theorem provides a handy criterion for the convergence in distribution of a sequences of HP functions. \\ 

\begin{theorem}\label{thm:suffcrit}
A sufficient condition for  a sequence  $ F_{\omega,n}  $ of random HP functions to converge in distribution to  a random HP function $F_\omega $ (i.e. for $\left(  \mu_{\omega,n},  F_{\omega,n}(i)   \right)  \stackrel{\mathcal{D}}{\to} \left(   \mu_{\omega} , F_{\omega}(i)\right) $) is that:
\begin{enumerate}
\item the corresponding random spectral measures $\mu_{\omega,n} $ converge in distribution to the random spectral measure $\mu_\omega$  in the sense that for all $ f \in C_c(\R) $:
\be \label{eq:cond1}
 \mu_{n}(f) := \int f(x)  \mu_{n}(dx) \stackrel{\mathcal{D}}{\to}  \mu(f) 
 \ee
\item there exists $ \Gamma \in \C^+ $ such that for all $ \varepsilon > 0 $: 
\begin{align}
\lim_{\eta \to \infty} \mathbb{P}(   \left| F(i\eta) - \Gamma\right| \geq \varepsilon ) & = 0 \, .  \label{eq:ass2}  \\ 
\lim_{\eta \to \infty} \limsup_{n \to \infty} \mathbb{P}(   \left| F_n(i\eta) - \Gamma\right| \geq \varepsilon ) \ & = 0   \, ,  \label{eq:ass1}  
\end{align}  

\end{enumerate}
\end{theorem} 

\begin{proof}  
We will write
\begin{align}\label{eq:split} 
 F_{\omega,n}(i) =  F_{\omega,n}(i\eta) +  \int \left( \frac{1}{x-i} - \frac{1}{x-i\eta} \right) \mu_{\omega,n}(dx)  =:  F_{\omega,n}(i\eta)  + \int g_\eta(x)  \mu_{\omega,n}(dx) \, . 
\end{align}
In a first step we establish that distributional convergence of the pair
\be\label{eq:firststep}
\left(  \mu_{\omega,n}\, ,  \int g_\eta(x)  \mu_{\omega,n}(dx) \right)  \stackrel{\mathcal{D}}{\to} \left(  \mu_{\omega} \, ,  \int g_\eta(x)  \mu_{\omega}(dx)\right)  
\ee
for all $ \eta \in [1,\infty ) $. To do so, we  split the integral into two parts
by inserting a smooth indicator function $ \chi_W \in C_c(\R) $ of the interval $ |x | \leq W $ with the property that 
$\chi_W^c(x) := 1 -  \chi_\eta(x) = 0 $ for all $ |x | \leq W $.  The pair $ \left(  \mu_{\omega,n} \, , \int g_\eta(x) \chi_W(x)   \mu_{\omega,n}(dx) \right) $ converges in distribution by assumption. 
The contribution to the integral from $ |x| \geq W $ is bounded:
\be\label{eq:boundIm}
\left| \int \chi_W^c(x) g_\eta(x) \mu_{\omega,n}(dx)\right| \leq \eta \int_{|x|\geq W} \frac{ \mu_{\omega,n}(dx)}{\sqrt{x^2+1}\sqrt{x^2+\eta^2}}  \leq \frac{2\eta}{W} \Im F_{\omega,n}(iW)\, . 
\ee
Choosing $ W = \eta^{1+\alpha} $ with some $ \alpha > 0 $, assumption~\eqref{eq:ass1} ensures that for any $ \varepsilon > 0 $:
\be
\lim_{\eta \to \infty} \limsup_{n\to \infty} \mathbb{P}\left( \left| \int \chi_{W_\eta}^c(x) g_\eta(x) \mu_{\omega,n}(dx)\right| \geq \varepsilon \right) = 0 \, .
\ee
This establishes~\eqref{eq:firststep}.  

The second assumption allows to convert \eqref{eq:firststep} to the statement  that the pair $( \mu_{\omega,n} , F_{\omega,n}(i) ) $ is asymptotically close (in distribution) to  $( \mu_{\omega}, F_{\omega}(i)  ) $, since the extra terms $F_{\omega,n}(i\eta) $ and $F_{\omega}(i\eta) $ are asymptotic (in probability) to the same constant $\Gamma$. 

This finishes the proof of the distributional convergence in the sense discussed in Section~\ref{sec:pointwise_conv}.   
\end{proof}

\subsection{Convergence of the distribution of the boundary values} 

To follow up on the  convergence criterion of Theorem~\ref{thm:suffcrit}, more needs to be said to address the convergence of the distribution of the random function along the boundary $\R $.   Following are some useful criteria, which will allow to apply the analysis of this paper to a number of cases of interest. \\  
 
\begin{theorem} \label{thm:conv_along_R}
 Let $ F_{\omega,n} $ be a sequence   of random HP functions  which converges in distribution, to a random function $ F_{\omega}$ (the sense discussed in Section~\ref{sec:pointwise_conv}), and suppose that in an interval $[a,b]\subset \R$    the spectral measures of  both $ F_{\omega,n}$ and $ F_{\omega}$  consist only of simple point processes.  Then also 
\be \label{eq:conv_at_R1}
F_{\omega,n}(x+i0)  \stackrel{\mathcal D}{\rightarrow} F_\omega(x+i0) \,  
 \ee 
 for any $x\in (a,b)$ for which  
 \be \label{eq:nocharge} 
 \Ev{\mu(\{x\})}=0  \, . 
   \ee 
 \end{theorem}  
 
 \noindent {\bf Remark:}  There is a reason here for the restriction on the nature of the spectra:  \eqref{eq:conv_at_R1} fails when the spectral measures of $F_{\omega,n}$ are discrete but converge to an absolutely continuous measure.  In such case there will be a positive measure set of $x\in \R$ at which $\Im F_\omega(x+i0) >0 $, while $\Im F_{\omega,n}(x+i0) =0 $ for all $n<\infty$. \\

\begin{proof} [Proof of Theorem~\ref{thm:conv_along_R}]
 In  the proof we  decompose each HP function $F_n(z)$ to a sum of two components, one due to  the {\it near} part of the spectral measure $\mu_F(du)$ and the other due to its {\it far} part.      
The distributional continuity  \eqref{eq:conv_at_R1}  of  the first component is where the limitation to discrete spectra is being used.  
This condition however  places the statement within the reach of standard continuity arguments.   The second contribution  is continuous by  Lemma~\ref{lem:hole}.  In addition to the separate continuity statements one needs  to notice that we have here a {\it  joint}  distributional convergence of the two components. \\ 
 
 Due to the freedom to shift and scale the result, it suffices to prove the assertions for the case $[a,b] = [-1,1]$, and sites 
 $x\in [-1/2,1/2]$.  Focusing on that case, let  $\chi : \R \mapsto [0,1]$ be the interpolated projection onto $[-1,1]$: 
\be
 \chi(x) \=  
 \begin{cases}  
     1 \quad  & \mbox{for  $|x|<1$} \\  
     1-2(|x|-1) & \mbox{ if $1<|x|<1.5$}\\ 
     0 & \mbox{for $|x|> 1.5$} 
\end{cases} 
\ee 
Using it, for each measure $\mu \in \mathcal M(\R)$ we denote its ``near'' and ``far'' parts as: 
\be 
\mu^{(1)}  (dx) \  := \  \chi(x) \, \mu   (dx)\, , \quad \mbox{and 
$\mu^{(2)}  (dx) \  := \  [1-\chi(x) ]\, \mu  (dx)$}  \, . 
\ee 
Correspondingly, we decompose any HP function $F(z)$ into: 
\be 
F(z) \= F^{(1)}(z)  \+ F^{(2)}(z)
\ee 
breaking the spectral representation  \eqref{eq:F} into: 
\begin{eqnarray}  
F^{(1)}(z)  &=&   \int_{[-1.5,1.5]} \left[ \frac{1}{u-z} - \frac{u}{u^2+1} \right] \, \chi(x) \, \mu^{(1)}_F(du) \notag \\ 
\mbox{ } \\ 
F^{(2)}(z)  &=&   b + a z + \int_{\R\backslash [-1,1]} \left[ \frac{1}{u-z} - \frac{u}{u^2+1} \right] \, [1-\chi(x)] \, \mu^{(1)}_F(du) \, .   \notag \\ 
\end{eqnarray}

It is easy to see that the assumed convergence in distribution of $F{\omega,n}$ implies the joint convergence of their two components as a pair of HP functions, in the natural extension of this notion to pairs of functions:  
 \be \label{eq:joint} 
\left(F^{(1)} _{\omega,n}, F^{(2)} _{\omega, n} \right )   \stackrel{\mathcal D}{\longrightarrow} \left( F^{(1)} _{\omega}, F^{(2)} _{\omega} \right )  \, . 
 \ee  
(In essence: the corresponding spectral measures converge for each value of $j$, and since $F^{(1)}$ falls off at infinity  the second requirement for convergence is of relevance only for  $j=2$.) \\

The assumed structure of  the spectral measure of $F^{(1)}_{\omega,n}(z) $  within $[-1,1]$ means that for each $n$ these random measures corresponds to a random  probability distribution  on the disjoint union of compact sets $Y \ := \ \cup_{k=0,1,2,...} [-1,1]^k$, the  point of each we shall denote by   $y_k=(y_{k,j})_{j=1}^k$, with  measures $\nu_{\omega,n}(d y_k)$ (on $k$ labeled particles) which are symmetric under permutations.    In particular, 
the probability that there are $k$  particles in $[-1,1]^k$ is   
\be 
p_n(k) \ := \ \int_{[-1,1]^k}   \nu_{\omega,n} (d\, y_k)/k!
\ee 
In this notation:  
\be \label{eq:local_F} 
F^{(1)}_{\omega,n}(z) \=   \sum_{k=1}^\infty \int_{[-1,1]^k}  \chi(y_j) \left[ \frac{1}{z-y_j} - \frac{y_j}{y_j^2+1} \right] \, \nu_{\omega,n} (d\, y_k)/k!\,.      \\ 
\ee 

In the natural topology on $Y$, the number of particles in $[-1,1]$ may change discontinuously due the appearance or disappearance of a particle at the boundary of the set.  Otherwise, the configuration depends continuously on the position of the particles in  $[-1,1]$.  Thus functions of the form 
\be \sum_{k=1}^\infty \int_{[-1,1]^k}  \phi(y_j)  \, \chi(y_j)  \nu_{\omega,n} (d\, y_k)/k! 
\ee 
with $\phi \in C([-1,1]) $ whose supported lies in $(-1,1)$   are continuous.   \\

Under the assumption of convergence in distribution of the random spectral measures, the sequence of probability measures $p_n$ on $\N$ is tight, and the integrals of  functions which are continuous in $[-1,1]^k$ and vanish at the boundary  have distribution which converges to that of the limiting measure.  
By the  {\it continuous mapping theorem}, this extends to functions which are continuous on a complement of a set which is not charged by the limiting  measure.   
In the representation \eqref{eq:local_F} of $F_n(x+i0) $ 
for a given $x$ the  integrand is a continuous function of the configuration except at configurations with a particle at $x$.    Thus,  the assumed convergence of the spectral measure allows to deduce the continuity of the probability distribution of $F^{(1)}_{\omega,n}(x+i0)$ for sites $x\in [-1/2,1/2]$ at which \eqref{eq:nocharge} holds.   \\

 The probability distribution of  $F^{(2)}_{\omega,n}(x+i0)$ is continuous in the limit $n \to \infty$ by an application of Lemma~\ref{lem:hole}.    Furthermore, combined with \eqref{eq:joint}, the  arguments imply that the joint distribution of the pair of random variables  
 $\left( F^{(1)}_{\omega,n}(x+i0), F^{(2)}_{\omega,n}(x)  \right)$ is continuous in the limit, and hence  \eqref{eq:conv_at_R1} holds.  \\ 
 \end{proof}  
 
Theorem~\ref{thm:conv_along_R} has implications for the random matrix models which are discussed next, and for the \v{S}eba process~\cite{Seba90, KMW,BBK,AlSeba}  on which more is said in \cite{ASW}.    Following is another continuity criterion which may be of interest beyond the cases covered by it, in particular when the spectral measures are singular but with dense support and not of uniform masses.  
An example to  keep in mind are the possible scaling limits of the Green functions of random operators in the regime of Anderson localization.  \\

In discussing continuity of HP functions along the line  it is  natural  to regard the range of $F$ as the Riemann sphere $ \overline {\C} $, i.e. the one point compactification of $\C$.   This suggests the following terminology.

\begin{definition} ($*$-continuity)  1. A function $F: \overline {\C} \mapsto \C$ is said (here) to be  $*$continuous at $z$ iff the mapping $z\mapsto \frac{-1}{F(z+i0) + i} $ is continuous at that point.

2.  For a random HP function, we define as its (mean) 
 {\it  modulus of $*$continuity}, at $x\in \R$,  the function 
\be 
  \kappa(x,\delta) \ :=  \  \Ev{ \left | \frac{1}{F(x+\delta+i0)+i} - \frac{1}{F(x+i0)+i} \right| } \,. 
\ee 
(which for almost all $x \in \R$ is defined for almost all $\delta\in\R $.) 
\end{definition} 
 
 \begin{theorem} \label{thm:conv_along_R*}
 If  a sequence   of random HP functions converges in distribution, $ F_{\omega,n}  \stackrel{\mathcal D}{\rightarrow} F_\omega $, and  the moduli of $*$-continuity of $F_{\omega,n} $ and $F_\omega$ are bounded uniformly in $ n $ by  $\kappa(x,\delta)$,  then  for any $x\in \R$ for which: 
 \be \label{eq:kappa_0}
 \lim_{\delta \to 0}  \kappa(x,\delta)  \= 0 
  \ee 
the distributions of the random variables $F_{\omega,n}(x+i0)$ converge:  
 \be \label{eq:conv_at_R}
F_{\omega,n}(x+i0)  \stackrel{\mathcal D}{\rightarrow} F_\omega(x+i0) \, . 
 \ee 
 \end{theorem}

 \begin{proof} 
 We already know, under the theorem's first assumption, that for each $x\in R$ and $\delta >0$: 
  \be
F_{\omega,n}(x+i\delta)  \stackrel{\mathcal D}{\rightarrow} F_\omega(x+i\delta) \, . 
 \ee 
To related this to the values at $\delta=0$, we note that by the Cauchy integral formula, for each $\delta >0$: 
 \be 
 \frac{1}{\pi} \int  \frac{1}{F_{\omega,n}(u+i0) +i} \; \frac{\delta \, du}{(x-u)^2+\delta^2}  \=  F_{\omega,n}(x+i\delta)  \, , 
 \ee 
 with similar relation holding for the limiting function $F_\omega$.   
 The difference can by estimated  in the $L^1$-sense by: 
 \begin{multline} \Ev{ \left|  \frac{1}{F(x+i0) +i}  \ - \ \frac{1}{\pi} \int  \frac{1}{F(u+i0) +i} \; \frac{\delta \, du}{(x-u)^2+\delta^2} \right|} \\  
  \le  \   \frac{1}{\pi} \int  \kappa(x, u-x) \, \frac{\delta \, du}{(x-u)^2+\delta^2} \  =: \  \hat \kappa(x,\delta)  \qquad 
 \end{multline}   
Under the assumption \eqref{eq:kappa_0} also:  $\hat \kappa(x,\delta) \to 0$ as $ \delta \to 0 $. 
Thus, a standard {\it three step comparison} allows to conclude the distributional convergence $ (F_{\omega,n}(x+i0) + i)^{-1} \stackrel{\mathcal D}{\rightarrow} (F_{\omega}(x+i0) + i)^{-1}  $ and hence the claim \eqref{eq:conv_at_R}. 
\end{proof}

\subsection{Examples from RMT and random operators}  \label{sec:examples}
 
 The above criterion can be verified for  the rescaled trace functions defined in 
\eqref{eq:FTrace}  for random matrices  corresponding to the two  examples which were discussed in Section~\ref{eq:St_examples}, whose spectra are  rather different. \\  

\noindent{\bf GUE and Wigner ensembles}~~The spectra of  $n \times n$ hermitian matrices  with complex Gaussian entries, which form the GUE random Gaussian ensemble, are well known to have for $n\to \infty$  the asymptotic density 
\be
  \varrho_{sc}(E_0) := \pi^{-1} \sqrt{ 1 - (E_0/2)^2}  \,  .  
 \ee
It is also known that the rescaled eigenvalue point process, 
amplified in the vicinity of energy $ |E_0| < 2 $, 
\be
 \mu_{\omega,n} = \sum_j \delta_{n  \,  \varrho_{sc}(E_0)  \,  (E_{j,n}(\omega) -E_0 ) } \, \,  , 
 \ee 
 converges in distribution to  the ``sine-kernel process'', which is 
a shift invariant determinantal point process $ \mu_\omega $ of kernel $ K(x,y) = \sin(  \pi (x-y)) / (\pi (x-y)) $   (cf.~\cite{Cupbook}).    \\

 In celebrated works \cite[Thm.~1.3]{EY12}, \cite[Thm.~5]{TV11} the above statement was recently generalized to the broader  class of  Wigner matrices, which are random hermitian $ n \times n $  matrices whose entries
  $ \{h_{jj} $,  $ \{ \Re h_{jk} \}_{j<k}$,  and $\{ \Im h_{jk}  \}_{j<k}$  are  independent, centered and of variance $ 1 /2 $.  
The quoted results imply that in the above case the rescaled trace function (cf.~\eqref{eq:trace})
\be\label{eq:rescaledtrace}
F_{\omega,n}(z) := \int \frac{ \mu_{\omega,n}(dx)}{x-z} =  \frac{1}{\varrho_{sc}(E_0) } \; R_{\omega,n}\left(  {E_0 +} \frac{z}{n \, \varrho_{sc}(E_0) }\right) \, , 
\ee
satisfies the first condition of Theorem~\ref{thm:suffcrit}, i.e. \eqref{eq:cond1} holds. \\

Of the criterion's second condition,
 \eqref{eq:ass2}, holds for the shifted random sine-kernel Stieltjes function $ F^{GUE}_\omega(z) + \Re \Gamma $ (cf.~\eqref{eq:GUE} and Subsection~\ref{sec:Stieljes}) with: 
\begin{equation}\label{def:GammaGUE}
 \Gamma :=   \frac{1}{\varrho_{sc}(E_0) } \int \frac{\varrho_{sc}(v)\, dv}{ v- E_0 - i0}  = -\frac{E_0}{2\varrho_{sc}(E_0)} + i \pi  \, . 
\end{equation}
The assertion that~\eqref{eq:ass1} holds also in the generality of Wigner matrices, of distributions with subgaussian tails, is implied by the statement derived in~\cite[Theorem~3.1]{ESY08} that at this generality, for all small enough $\varepsilon > 0 $: 
\be
\lim_{\eta \to \infty} \limsup_{n\to \infty}\mathbb{P}\left( \left| F_{\omega,n}(i\eta)  -\Gamma \right| \geq  \varepsilon \right) \leq \lim_{\eta \to \infty} C e^{- c\varepsilon\sqrt{\eta}} = 0 \,  
\ee
at  some $ c, C < \infty$ (while this suffices for our purpose, an improved bound was recently presented in~\cite{CMS13}).  \\

Combining these statements with the general criterion provided by Theorem~\ref{thm:conv_along_R}, 
one gets\footnote{L. Erd\"os and A. Knowles also noted  that  such conclusion 
may be drawn from our Theorem~\ref{thm:Cauchy} combined with previous RMT analysis, basing their argument on the more recent  results of~\cite{CMS13} followed by some additional analysis.}:
\begin{corollary} \label{cor:RMT}   
For Wigner matrices $H_{\omega,n} $ whose entries have  a common subgaussian distribution $ \nu $, i.e., $ \int e^{\delta x^2} \nu(dx) < \infty $ for some $ \delta > 0 $, the rescaled trace $ F_{\omega,n}(x) $, defined by~\eqref{eq:rescaledtrace},  converges  in distribution, for  $ n \to \infty $ and any fixed $x$,  to a Cauchy random variable whose analytic baricenter $ \Gamma $ is given by~\eqref{def:GammaGUE}.
\end{corollary}

\medskip 
 
\noindent {\bf Random diagonal matrices.}~~A similar  statement is valid also for the much simpler ensemble of 
  $ n\times n $ random diagonal matrices, whose diagonal entries $(V_j)$ are of a common probability distribution with a smooth density $\rho \in C^1(\R) $.   In this case the rescaled trace function
\be
 F_{\omega,n}( z) : =  \sum_{j=1}^n \frac{1}{ n \, \rho(E_0)  [V_j - E_0] - z} = \frac{1}{\rho(E_0) } \; R_{\omega,n}\left({E_0 +} \frac{z}{n \, \rho(E_0) }\right) \, , 
\ee
with $ E_0$ such that $ \rho(E_0) > 0 $, converges in distribution for any $ z \in \overline{\C^+} $ to the shifted Poisson-Stieltjes function  $ F^{Poi}_\omega(z)  + \Re \Gamma $ with
\be\label{eq:Gamma3}
\Gamma := \frac{ 1}{\rho(E_0)} \int \frac{\rho(v)}{v-E_0- i 0} dv = \frac{1}{\rho(E_0)} \, P.V. \int \frac{\rho(v)\, dv}{v-E_0}   + i \pi \, .  
\ee
In particular, for any $ x \in \R $ the random variables $  F_{\omega,n}( x) $ converge in distribution as $ n \to \infty $ to a Cauchy random variable with baricenter $ \Gamma $ given by~\eqref{eq:Gamma3}.\\ 

Here, the assertion can be easily proven by a direct computation of the characteristic functional $ \mathbb{E}\left[ e^{it F_{\omega,n}( x) } \right] $.   Alternatively, it also follows from Theorems~\ref{thm:suffcrit} and~\ref{thm:conv_along_R} and the fact that $ F^{Poi}_\omega(x) $ has a Cauchy distribution with baricenter $ (i \Im \Gamma)$, cf.~Theorem~\ref{thm:Stieltjes} and~\ref{thm:Cauchy}.

\appendix

\section{Boole's identity for the Stieltjes transform of singular measures} \label{app:Boole}

Following is the proof of  Proposition~\ref{prop:boole},  which we assume is known to experts.  For convenience we restate the result, which extends an identity of Boole \cite{Boole} from the case of pure-point measure $ \mu $ to general singular measures. 

\begin{theorem}   
Let  
\be 
F(z)\= \int \frac{\mu_F(du)}{u-z} 
\ee 
 with the spectral measure $\mu_F$ which is finite and purely singular with respect to the Lebesgue measure $\mathcal L$ (or equivalently:  $\Im F(x+i0) =0$ for a.e. $x\in \R$). Then for any $ t >0 $:
\begin{equation}
	\mathcal L \left( \left\{ x \in \R \, | \,  F(x+i0)  \geq t \right\} \right)  \ = \ \frac{\mu_F(\R)}{t} \, . 
\end{equation}
\end{theorem}
\begin{proof}
	The monotone convergence theorem implies that 
	\begin{equation}
		\mathcal L \left(  \left\{ x \in \R \, | \,   F(x+i0)  \geq t   \right\} \right) \ = \ \lim_{\eta \to \infty } \int \frac{\eta^2}{x^2 + \eta^2} \, \indfct[  F(x+i0)  \geq t]  \, dx \, . 
	\end{equation} 
	 The proof is  based on the observation that the distribution of the random variable $ F(x+i0) $ with respect to the Cauchy probability measure $\frac{\eta}{ x^2 + \eta^2} \,  \frac{dx}{\pi} $ is uniquely characterized by its characteristic function, which by contour integration is:
	 \begin{equation}\label{eq:char}
	 	\int e^{i \tau F(x+i0)} \frac{\eta}{ x^2 + \eta^2} \,  \frac{dx}{\pi} = e^{i\tau \Re F(i\eta)} \, e^{- |\tau | \,  \Im F(i\eta) }  \, ,
	 \end{equation}
where the integral was evaluated for $ \tau > 0 $ using a contour integration argument using the analyticity of $ F$ in the upper half plane $ \C^+ $. For $ \tau < 0 $ the characteristic function is obtained by complex conjugation from the one for $ \tau > 0 $ (since $F(x+i0)$ is real).  Equation~\eqref{eq:char} shows that with respect to the Cauchy probability measure, the distribution of the variable $ F(x+i0) $ is itself Cauchy centered at $ \Re F(i\eta) $ of width $  \Im F(i\eta)  $. As a consequence,
	 \begin{align}
	 	\lim_{\eta \to \infty } \int \frac{\eta}{x^2 + \eta^2} \, \indfct[  F(x+i0)  \geq t]  \, dx \ &  = \   \lim_{\eta \to \infty} \frac{\eta \, \Im F(i\eta) }{ t - \Re F(i\eta)} 
			\ = \ \frac{\mu_F(\R)}{t} \, , 
	 \end{align}
	 since $ \lim_{\eta \to \infty} \Re  F(i\eta) =  \lim_{\eta \to \infty}   \Im F(i\eta)= 0 $ and $  \lim_{\eta \to \infty}  \eta \, \Im F(i\eta)  = \mu_F(\R) $. 
\end{proof}

%
%
%
%
%

\section*{Acknowledgments} 
We thank A. Knowles, L. Erd\"os, Y.V. Fyodorov and O. Zeitouni for relevant comments.    
M. Aizenman was supported in parts by the NSF grant PHY-1104596      
and by the Weston Visiting Professorship at  the Weizmann Institute of Science;  S. Warzel  was supported in part by the von Neumann Visiting Professorship  at the Princeton Institute for Advanced Study.  
%

 \end{document}